\documentclass{stacs_proc}

\title{Deterministically Isolating a Perfect Matching in Bipartite Planar 
Graphs}
\author[ref1]{S. Datta}{Samir Datta}
\address[ref1]{Chennai Mathematical Institute, Chennai, India.}
\email{sdatta@cmi.ac.in}

\author[ref2]{R. Kulkarni}{Raghav Kulkarni}
\address[ref2]{University of Chicago, Chicago, USA.}
\email{raghav@cs.uchicago.edu}
\thanks{This work was done while the second author was visiting Chennai Mathematical Institute}

\author[ref3]{S. Roy}{Sambuddha Roy}
\address[ref3]{ IBM Research Laboratory, New Delhi, India.}
\email{sambuddha@in.ibm.com}

\usepackage{amsmath,amssymb,latexsym}
\usepackage{graphicx, epsfig}
\theoremstyle{plain}\newtheorem{claim}[thm]{Claim}

\newcommand{\Log}{\mbox{{\sf L}}}

\newcommand{\NL}{\mbox{{\sf NL}}}
\newcommand{\UL}{\mbox{{\sf UL}}}

\newcommand{\UPM}{\mbox{{\sf UPM}}}

\newcommand{\Pt}{\mbox{{\sf P}}}

\newcommand{\SPL}{\mbox{{\sf SPL}}}
\newcommand{\NC}{\mbox{{\sf NC}}}

\begin{document}
\maketitle

\stacsheading{2008}{229-240}{Bordeaux}
\firstpageno{229}

\begin{abstract}
We present a deterministic way of assigning
small (log bit) weights to the edges of a bipartite planar
graph so that the minimum weight perfect matching becomes 
unique. The {\em isolation lemma} as described in 
\cite{mvv87} achieves the same for general graphs using
a randomized weighting scheme, whereas we can do it 
deterministically when restricted to bipartite planar graphs.
As a consequence, we reduce both decision and construction versions
of the matching problem to testing whether a matrix is
singular, under the promise that its determinant is $0$ or $1$,
thus obtaining a highly parallel {\SPL} algorithm
for bipartite planar graphs. This improves the
earlier known bounds of non-uniform \SPL\ by \cite{arz99} and
\NC$^2$ by \cite{mn95,mv00}. It also rekindles the hope
of obtaining a deterministic parallel algorithm for 
constructing a perfect matching in non-bipartite planar graphs, 
which has been open for a long time. Our techniques are
elementary and simple.
\end{abstract}
\section {Introduction}
The {\em Matching Problem} is one of the most well-studied in the field of 
parallel complexity. Attempts to solve 
this problem have led to the development of a variety of combinatorial, algebraic and 
probabilistic tools which have applications even outside the field. 
Since the problem is still open, researchers linger
around it in search of new techniques, if not to solve it in its whole 
generality, then at least under various natural restrictions. 
In this paper, we will focus on the deterministic complexity of the 
{\em Matching Problem} under its planar restrictions.
 
\subsection {The Matching Problem} 
\begin {definition}
A matching in an undirected graph is a collection of edges which have no 
endpoint in common. 
\end {definition}
Such a collection of edges is called ``independent".
See \cite{lp86} for an excellent survey on matchings.

The computational question one can ask here is, given
a graph, to find a matching of the maximum cardinality.

\begin{definition}
A perfect matching in a graph is a collection of independent edges which cover
all the vertices.
\end{definition}
One may ask various computational questions about
perfect matchings in graphs. We will consider the following three questions: \\
{\em Question 1:} (Decision) Is there a perfect matching in a given graph ? \\
{\em Question 2:} (Search) Construct a perfect matching in a graph,
if it exists. \\
{\em Question 3:} (Uniqueness Testing or \UPM) Does a given graph have exactly one perfect matching? \\
There are polynomial time algorithms for the above graph matching problems
and historically people have been interested in studying the parallel 
complexity of all the three questions above. The \UPM\ question for
bipartite graphs is deterministically parallelizable \cite{kvv85}
(i.e. it lies in the complexity class \NC; see any standard complexity
text for a formal definition, say \cite{v99}). 
 Intuitively, \NC\ is a complexity class consisting of the problems having a
parallel algorithm which runs in polylogarithmic time using polynomially 
many processors which have access to a common memory. 

It is the class consisting of so called ``well parallelizable''
problems. \NC\ is inside \Pt\ - problems having
a sequential polynomial time algorithm. 
Whether the {\em Matching Problem} is
deterministically parallelizable remains a major open question in 
parallel complexity. 

\begin{oprob} {Is Matching in \NC\ ?}
\end {oprob}

The best we know till now is that {\em Matching} is in {\em Randomized \NC}.
See for example, \cite{kuw86,mvv87}.
Several restrictions of the matching problem are known to be in \NC,
for example, bipartite planar graphs \cite{mn95,mv00}, graphs with
polynomially bounded number of perfect matchings \cite{gk87} etc.
Whether the search version reduces to the decision version has also not been answered yet.

\subsection {Randomized Isolation Lemma }
\begin {lemma} \cite{mvv87} One can randomly assign polynomially bounded 
weights to the edges of a graph so that with high probability the minimum
weight perfect matching becomes unique.
\end{lemma}

Using the isolation lemma, \cite{mvv87} obtained a simple {\em Randomized \NC}
algorithm for finding a perfect matching in arbitrary graphs.

\subsection {Matching in \SPL/poly }
Allender et al \cite{arz99} proved a non-uniform bound for matching problem 
which
allows us to replace the randomization by a polynomial length advice string.
Hence, we know that matching is parallelizable with polynomial bit advice.

\begin{definition} 
\SPL\ is a promise class that is characterized by the problem of checking 
whether a a matrix is singular under the promise that its determinant is
either $0$ or $1$. The corresponding non-uniform class
\SPL/poly is \SPL\ with a polynomial bit advice.
\end{definition}

\SPL\ is inside $\oplus \Log$ and inside $\oplus_p \Log$ for all $p.$ While
\UL\ (unambiguous Logspace) is inside {\SPL}, \NL\ (nondeterministic
Logspace) is incomparable with \SPL.  Both \NL\ and \SPL\ are known to 
lie inside $\NC^2.$

\begin{definition}
A language is said to be in \SPL/poly if for every positive integer $n$ there
exists an advice string $A_n$ such that:
\begin{itemize}
\item length of $A_n$ is polynomially bounded in $n$ 
\item once $A_n$ is given, the membership of any input of size $n$
can be decided in \SPL.
\end{itemize}
\end{definition}

\begin{theorem} \cite{arz99} Matching is in \SPL/poly.
\end{theorem}


\subsection {Matchings in Planar Graphs and Deterministic Isolation}
The situation for planar graphs is interesting because of the 
fact that counting the number of perfect matchings in planar graph
is in \NC\ (\cite{k67,v88}) whereas constructing one perfect matching is not
yet known to be parallelizable. However,
for bipartite planar graphs, people have found \NC\ algorithms
\cite{mn95,mv00}. 

The isolation lemma crucially uses randomness in order to isolate a
minimum weight set in an arbitrary set system. It is conceivable,
however, to exploit some additional structure in the set system
to eliminate this randomness.  Indeed, recently \cite{btv07} building upon
a technique from \cite{adr05} were able to 
isolate a directed path in a planar graph by assigning small 
{\em deterministic} weights to the edges. The lemma that sits at the
heart of that result says that there is a simple deterministic way to assign
weights so that each directed cycle (in a grid graph)
gets a non-zero weight. This is shown to imply 
that if two paths get the same weight neither of them is a min-weight path.

Motivated by their result we explore
the possibility of such an isolation for perfect matchings in planar graphs.
Our attempt is to assign weights so that the alternating sum is non-zero
for each alternating cycle - here alternating sum is the signed sum of weights
where the sign is opposite for successive edges. Since alternating cycle result
from the super-imposition of two matchings, we are able to isolate a min-weight
matching.

Therefore, we are able to devise an \NC\ algorithm for bipartite planar 
graphs which
is conceptually simple, different from the other known algorithms
and tightens its complexity to the smaller class \SPL.
The search problem for matching in non-bipartite planar graphs still remains
open even though the corresponding decision and counting versions
are known to be in \NC. Our algorithm rekindles the hope for solving
general planar matching in \NC.

\section {Preliminaries}
Here we describe the technical tools that we need in the rest of the paper.
Refer to any standard text (e.g. \cite{v99}) for definitions of the
complexity classes $\oplus \Log$, $\oplus_p \Log$, {\NL}, {\UL}, $\NC^2$ .
For graph-theoretic concepts, for instance, {\em planar graph}, {\em outerplanar graph},
{\em spanning trees}, {\em adjacency matrix}, {\em Laplacian matrix} of a graph, we refer the reader to any standard text in graph theory
(e.g. \cite{diestel}).

\subsection {Definitions and Facts}
We will view an $n \times n$ grid as a graph simply by putting the nodes at the 
grid points and letting the grid edges act as the edges of the graph. 
\begin{definition} 
{\em Grid graphs} are simply subgraphs of an $n \times n$ grid for some $n$.
See Figure~\ref{gridgraph} for an example.
We call each unit square of the grid a {\em block}.
\end{definition}
\begin{definition}
We call a graph an {\em almost grid graph} if it consists of a grid graph
and possibly some additional diagonal edges which all lie in some single
row of the grid.  Moreover all the diagonal edges
are parallel to each other. See Figure~\ref{almostgrid}. 
\end{definition}

In this paper we will consider weighted grid graphs where each edge is 
assigned an integer weight. \\

\begin{figure}
\includegraphics[scale = 0.25] {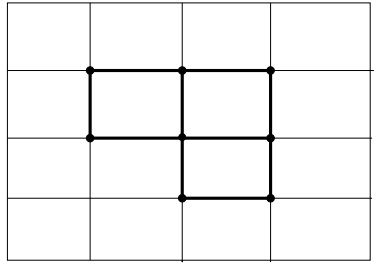}
\caption{A Grid Graph}
\label{gridgraph}
\end{figure}

\begin{figure}
\includegraphics[scale = 0.25] {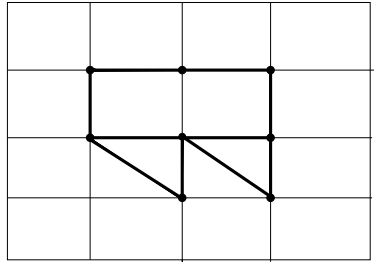}
\caption{An Almost Grid Graph}
\label{almostgrid}
\end{figure}

\begin{definition} 
\begin{enumerate}
\item Given a grid,  assign a ``+" sign to all the vertical edges 
and a ``-" sign to all the horizontal edges.
\item Assign a sign of $(-1)^{i+j}$ to the block in the
$i^{\mbox{th}}$ row and $j^{\mbox{th}}$ column (adjacent blocks get opposite signs).
\end{enumerate}
\end{definition}

\begin{definition}
Given a weighted grid graph $G$, the
{\em circulation} of a block $B$(denoted $circ(B)$) in $G$ is the signed sum of weights of the edges of it:
$circ (B) = \sum_{e \in B} {sign(e) weight(e)}$.
\end{definition}
\begin{definition} 
Given a weighted grid graph $G$ and a simple cycle 
$C = (e_0, e_1, \dots , e_{2k-1})$ in it, 
where $e_0$ is, say, the leftmost topmost vertical edge of $C$; 
we define the circulation of a cycle $C$ as 
$circ(C) = \sum_{i=0}^{2k-1} { (-1) ^ i weight(e_i)}.$
\end{definition}

The following lemma plays a crucial role in constructing 
non-vanishing circulations in grid graphs as will be described in the
next section.
\begin{lemma}{\em Block Decomposition of Circulations:}
The absolute value of the circulation of a simple cycle $C$ in a grid graph $G$ is 
equal to the signed sum of
the circulations of the blocks of the grid which lie in the interior of $C.$ \\
$|circ (C)| =  \sum_{B \in interior(C)} {sign(B)circ(B)}$
\label{blockdec}
\end{lemma}
\begin{proof}
 Consider the summation on the right hand side.
The weight of any edge in the {\em interior} of $C$ will get cancelled
in the summation because that edge will occur in exactly two blocks 
which are adjacent and hence with opposite signs.
Now what remains are the boundary edges. Call two boundary edges 
{\em adjacent} if they appear consecutively on the cycle $C$.

\begin{claim}\label{claim1} Adjacent boundary edges get
opposite signs in the summation on the right hand side above.
\end{claim}
\begin{proof}
We have to consider two cases, either the adjacent boundary edges
lie on adjacent blocks, in which case since adjacent blocks have
opposite signs, these edges will also get opposite signs as they are
both vertical or horizontal edges. See Figure~\ref{gridSignWeight}.
In the other case, when adjacent boundary edges do not lie on adjacent blocks,
they lie on two blocks which are diagonally next to each other.
In this case, both blocks will have the same sign but since one edge is vertical
and the other is horizontal, the effective sign of the edges will be opposite.
See Figure~\ref{gridSignWeight}.
Hence, the adjacent boundary edges will get opposite sign in the summation.
This completes the proof that the right hand side summation is precisely
+ $circ (C)$ or - $circ (C).$ 
\end{proof}
\end{proof}

\begin{figure}
\centering
\begin{tabular}{cc}
\epsfig{file=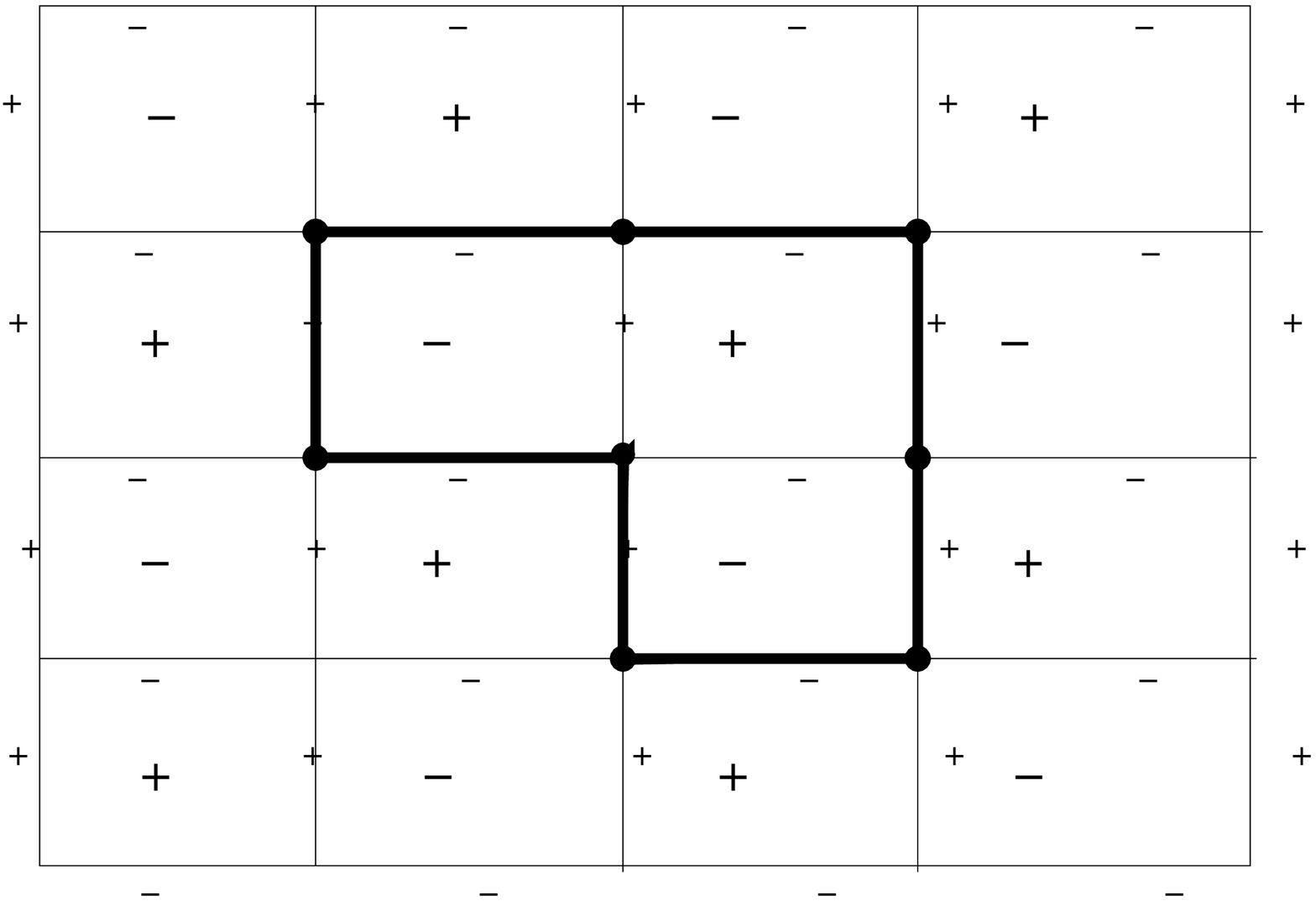,width=0.5\linewidth,clip=} &
\epsfig{file=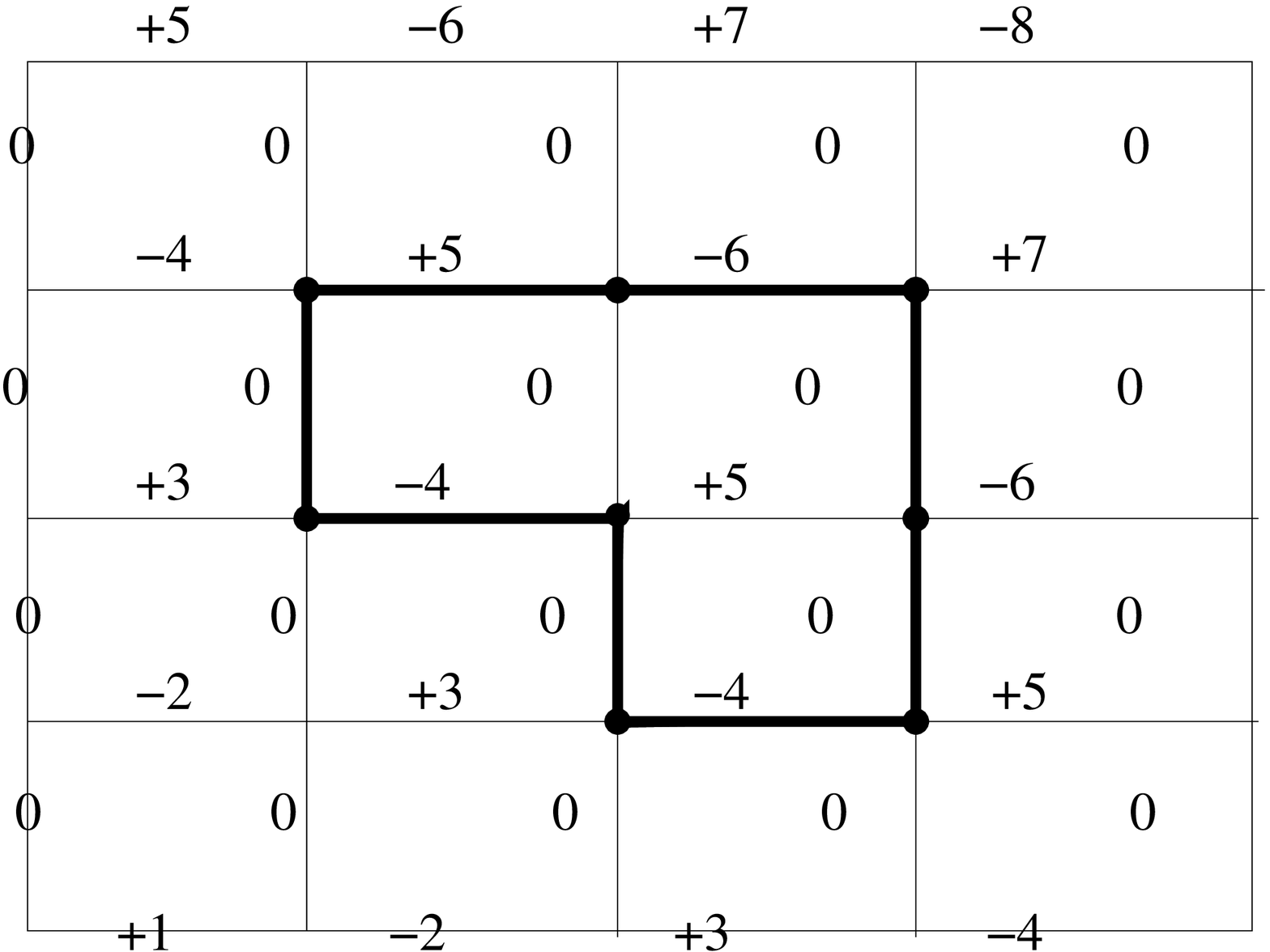,width=0.5\linewidth,clip=} \\
\end{tabular}
\caption{Signs and Weights of the blocks and the edges of a grid}
\label{gridSignWeight}
\end{figure}
We will also have occasion to employ the following lemma and we
record it here:
\begin{lemma} {\em Temperley's Bijection:} 
The spanning trees of a planar graph are in one to one correspondence with perfect
matchings in a bipartite planar graph. Moreover the correspondence is weight
preserving. 
\end{lemma}
This bijection was first observed by Temperley around 1967.
Recently \cite{kpw00} have found 
a {\em Generalized Temperley Bijection} which gives a one-to-one 
weight preserving mapping between directed rooted spanning trees 
or arborescences in a
directed planar graph and perfect matchings in an associated 
bipartite planar graph.

\subsection {Planar Matching and Grid Graphs}
Grid graphs have turned out to be useful for solving the reachability question 
in directed planar graphs, cf. \cite{abc06,btv07}. 
Motivated by this fact we explore the possibility of reducing planar
matching problem to that of grid graphs. Non-bipartiteness 
becomes an obstacle here which leaves us with  the following observations:
\begin {lemma} One can convert any bipartite planar graph into a 
grid graph such that the perfect matchings remain in
one-to-one correspondence.
\end {lemma}
\begin{proof} This is described in \cite{dkl07}.
It follows closely the procedure for embedding a planar
graph into a grid, described by \cite{abc06}. 
\end{proof}

Though non-bipartiteness is an issue, we can get rid of it to
a certain extent, though as expected, not completely .
\begin {lemma} Any planar graph, not necessarily bipartite,
can be converted to an {\em almost 
grid graph} while maintaining the one to one correspondence 
between the perfect matchings.
\label{almost}
\end {lemma}
\begin{proof} This procedure is analogous to the previous one except that we
can observe that the edges which are causing non bipartiteness can be 
elongated into a long path and placed in a grid so that only in a single row 
one needs to use a diagonal edge. 
\end{proof}

\section { Bipartite Planar Perfect Matching in \SPL }
In this section, we will give a simple algorithm for
finding a perfect matching in bipartite planar graphs, also improving
over its complexity by putting it in \SPL.
Earlier the best known bound was $\NC^{2}.$ See for example
\cite{mn95,mv00}.
At the core of our algorithm, lies a procedure to deterministically 
assign the small (logarithmic bit long)
weights to the edges of a bipartite planar graph, so that the minimum weight
perfect matching becomes unique. A simple observation about non-vanishing 
circulations in bipartite planar graphs makes it possible to isolate a 
perfect matching in the graph, which can be further extracted out using an 
\SPL\ query. 
\subsection {Non-vanishing Circulations in Grid Graphs}
We are interested in assigning the small weights to the edges of a grid
so that any cycle in it will have non-zero circulation.
This weighting scheme is at the heart of the isolation of perfect matchings
in grid graphs. The procedure runs in Logspace.
\begin{lemma} One can assign, in Logspace, small (logarithmic bit) weights to 
the edges of a grid so that circulation of any cycle becomes non-zero.
(One weighting scheme which guarantees 
non-zero circulation for every cycle in the grid is shown in 
the Figure~\ref{gridSignWeight}.)
\label{wt}
\end{lemma}
\begin{proof}
We assign all vertical edges weight $0$ and horizontal edge
from grid point $(i,j)$ to $(i+1,j)$ is assigned a weight of
$(-1)^{i+j}(i+j+1)$
as shown in figure ~\ref{gridSignWeight}. Thus the circulation of the block
with diagonally opposite vertices $(i,j)$ and $(i+1,j+1)$ is
$\sum_{e}sign(e)weight(e)$ $= (-1)(-1)^{i+j}(i+j+1) + (+1)0 + (-1)(-1)^{i+j+1}(i+j+2) + (+1)0$
$= (-1)^{i+j}$

Thus, the weighting scheme makes sure that the circulation of any block is either
+$1$ or - $1.$ Moreover, the circulation of a block is positive if and only if
its sign is positive. 
Now, using the Block Decomposition of Circulations (Lemma \ref{blockdec}), 
we have that the circulation of any cycle in absolute value is precisely 
the number of blocks in the interior of it, and hence is never zero.
\end{proof}
\subsection {Non-vanishing Circulations: A Direct Method}
One can think of the procedure of assigning the weights to the edges of 
bipartite planar graph without having to convert it to a grid graph. 
The procedure is as follows:
\begin{itemize}
\item {\bf Step 1:} {\em Make the graph Eulerian (every vertex has an 
even degree):}
add spurious edges to it without disturbing the bipartiteness.
\begin{itemize}
\item {\em Step 1.1:} 
Perform an Euler traversal on a spanning tree in the dual graph. 
\item {\em Step 1.2:}
While performing the traversal, make sure that when you leave
the face, all the vertices in the face, except for the end points of
the edge through which we go to the next face, are of even degree. 
To guarantee this we can do the following: 
\begin{itemize}
\item {\em Step 1.2 a) :} Start with one end point 
say $u$ of the edge $(u,v)$ through which we go to the next face. 
Visit all the vertices of the face in a cyclic ordering, every time connect  
an odd degree vertex to the next vertex by a spurious edge. 
\item {\em Step 1.2 b) :}
If the next vertex is also
of odd degree then go to its next vertex. 
If the next vertex is of even degree
then we have pushed the oddness one step further. 
\item {\em Step 1.2 c) :}
Continue the same procedure
till we remove all the oddness or push it to $v.$ 
\item {\em Step 1.2 d) :}
In the process, the graph might become a multi-graph i.e. between two nodes we 
may have multiple edges, but this can be taken care of by replacing every 
multi-edge by a path of length $3.$ The bipartiteness is preserved in the 
process.
\end{itemize}
\end{itemize}
\item {\bf Step 2:}{\em Fix the signs:} 
After Step 1, the graph has become Eulerian, and hence the dual graph 
is bipartite. 
\begin{itemize}
\item {\em Step 2 a)} {\em Assign alternating signs to adjacent faces:}
Form a bipartition of the dual, and fix all the faces in one bipartition to 
have + sign and the others to have - sign.
Any two adjacent faces will have opposite signs. 
Here, faces will act as blocks. 
\item {\em Step 2 b)} {\em Assign alternating signs to adjacent edges of 
every face:}
Consider an auxiliary graph obtained from the bipartite planar graph as 
follows: Every new vertex corresponds to an edge in the graph.  
Join two new vertices by a new edge 
iff the corresponding edges in the original graph are adjacent along some 
face. Now since both the original graph and its dual are bipartite , 
the auxiliary graph will also be bipartite - hence edges in the two bipartitions get opposite signs ensuring that around every face the signs are alternating.
\end{itemize}
\item {\bf Step 3:} {\em Assign small weights to the edges:} 
Now make another Euler traversal on the dual tree everytime assigning the weight
to the dual tree edge through which you traverse to the next face so that the 
circulation of the face you leave is exactly same as the sign of the face.
All non-tree edges will be assigned zero weight.
It is easy to see that all the weights assigned are small (logarithmic bit).
\end{itemize}

\noindent {\em Block Decomposition of Circulations:}
Again, the circulation of a cycle will decompose into signed sum of 
circulations of the faces in the interior of it and since
the sign and the circulation for any face is the same, we will have 
non-vanishing circulations in the graph. The details are analogous
to the case of a grid. We leave the details to the reader.

\subsection {Deterministic Isolation}
\label{isolation}
The non-vanishing circulations immediately give us the isolation 
for the perfect matchings.
\begin{lemma} If all the cycles in a bipartite graph have
non-zero circulations, then the minimum weight perfect matching in it is unique.
\end {lemma}
\begin{proof} If not, then we have two minimum weight perfect 
matchings $M_1$ and
$M_2$ which will contain some
alternating cycles, and each such cycle is of even length. Consider any one
such cycle. Since the circulation of an even cycle is nonzero
either the part of it 
which is in $M_1$ is lighter or the part of it which is in $M_2$ is lighter, 
in either case, we can form another perfect matching which is lighter than the 
minimum weight perfect matching, which is a contradiction.
\end{proof}
 
Thus we have a deterministic way of isolating a perfect matching
in bipartite planar graphs, and it is easy to check that the 
procedure of assigning the weights to the edges works in
deterministic Logspace.

\subsection {Extracting the Unique Matching}
Once we have isolated a perfect matching, one can extract it out easily
as follows:
\begin {itemize}
\item {\bf Step 1:} 
Construct an $n \times n$  matrix $M$ associated with a planar 
graph on $n$ vertices as follows: 
Find a {\em Pfaffian orientation} (\cite{k67})
of the planar graph and put the
$(i,j)$th entry of the matrix $M$
to be $x^{w_{(i,j)}}$ where $x$ is a variable and $w_{(i,j)}$ is the weight
of the edge $(i,j)$ which is directed from $i$ to $j$ in the Pfaffian 
orientation. If the edge is directed from $j$ to $i$ then
put $- x^{w_{(i,j)}}$ as $(i,j)th$ entry of the matrix.
\item {\bf Step 2:} 
If $t$ is the weight of the minimum weight perfect matching, then
the coefficient of $x^{2t}$ in determinant of $M$ will be the 
number of perfect matchings of weight $t.$ Now, as shown in 
\cite{mv97,v99} this coefficient can be written as a determinant of 
another matrix.
\item {\bf Step 3:} 
Now start querying from $i = - n^2$ to $+ n^2$ whether 
the coefficient of $x^{2i}$ is zero or not and
the first time you find that it is nonzero; stop. The first nonzero value
will occur when $i = t$ and it will be $1$ since the minimum weight 
perfect matching is unique.
Hence, during the process, every time we have a promise that if
the determinant is non-zero, it is $1.$ This procedure gives the
weight of the minimum weight perfect matching.
\item {\bf Step 4:}
Now once we know the procedure to find the weight of the minimum weight 
perfect matching, then one can find out which edges are in the matching by 
simply deleting the edge and again finding  the weight of the minimum weight 
perfect matching in the remaining graph. If the edge is in the the isolated
minimum weight perfect matching then after its deletion the weight of the new
minimum weight perfect matching will increase. Otherwise we can conclude that 
the edge is not in the isolated minimum weight perfect matching.  
Note that the isolation holds even after deleting an edge from the graph.
\end{itemize} 
\begin{theorem}\label{thm:main} Bipartite Planar Perfect Matching is in \SPL.
\end {theorem}
\begin{proof} The Logspace procedure in Lemma \ref{wt} assigns the 
small weights to the edges of the graph
so that the minimum weight perfect matching is unique and the above 
procedure extracts it out in $L^{\SPL} = \SPL.$
\end{proof}
We obtain the following corollaries.
\begin{corollary}\label{cor1} \UPM\ in bipartite planar graphs is in \SPL.
\end{corollary}
\begin{proof} To check whether the graph has unique perfect matching, one 
can construct one perfect matching and check that
after removing any edge of it there is no other perfect matching.
\end{proof}
\begin{corollary}\label{cor2} Minimum weight perfect matching in planar graphs 
with polynomially bounded weights is computable in \SPL.
\end{corollary}
\begin{proof} One can first scale the polynomially bounded weights
by some large multiplicative factor, say $n^4$ and then perturb them 
using the weighting scheme described above so that some minimum weight
matching with original weights remains minimum weight matching with new
weights and is unique. Then extraction can be done in \SPL.
\end{proof}
\begin{corollary} Minimum weight spanning tree in planar graphs is 
computable in \SPL\ if the weights are polynomially bounded.
(The same is true for directed rooted spanning trees (arborescences) in planar
graphs due to Generalized Temperley's bijection shown in
\cite{kpw00}.)
\end{corollary}
\begin{proof} This follows from Temperley's bijection.
\end{proof}

Restricting the family of planar graphs, yields better upper bounds for Matching
questions. Notably, we prove that:
\begin{corollary}\label{cor4}(of Theorem~\ref{thm:main}) Perfect Matching in outerplanar graphs is in \SPL.
\label{outerplanar}
\end{corollary}
\begin{proof} If we have an outerplanar(1-page) graph on $n$ vertices 
with vertices labelled from $1$ to $n$ along the spine, then 
observe that the edges between two odd labelled vertices can not be
part of any perfect matching. This is because the number of vertices below 
that edge is odd. Similarly edges between two even labelled vertices 
can not participate in any perfect matching. Hence, by removing such edges
we can make the graph bipartite and then we can apply the previous theorem.
\end{proof}

We use the lemma below in order to prove the theorem that follows it.
\begin{lemma}\label{lemma:parOP}
The parity of the number of perfect matchings in an outerplanar graph can be 
computed in Logspace.
\end{lemma}
\begin{proof}
 It is not hard to observe that the parity of the determinant of the adjacency matrix of a graph
is the same as the parity of number of perfect matchings in it. 
Finding the parity of the adjacency matrix of an outerplanar graph can be 
reduced to finding the parity of the number of spanning trees in an auxiliary
planar graph which is constructed by adding a new vertex and connecting it to
all the odd degree vertices of the original graph.
The new graph has all the vertices of even degree. Hence the adjacency matrix 
of the new graph is the same as its Laplacian modulo $2.$ Now the minor obtained
by removing the row and the column corresponding to the new vertex, 
is precisely the adjacency matrix of the original outerplanar graph modulo $2.$
Also the determinant(mod $2$) of this minor is precisely the parity of the 
spanning trees in new graph. 
As shown in \cite{bkr07}, the  parity of spanning trees modulo $2$
in planar graphs can be computed in Logspace.
Hence, the parity of the determinant of the adjacency matrix of an outerplanar graph 
can be obtained in Logspace which in turn gives the parity of the number of
perfect matchings in it.
\end{proof}

\begin{theorem}\label{thm:op} \UPM\ in outerplanar graphs is in Logspace.
\end{theorem}
\begin{proof}
If the perfect matching in an outerplanar graph is unique, one can 
obtain one perfect matching in Logspace. For every edge, one just needs to 
compute the parity of the number of perfect matchings in the graph with
that particular edge removed. If this parity is $1$ then don't include this
edge in the perfect matching, otherwise do.
Now we just need to verify that the perfect matching thus constructed is unique.
As seen in Corollary~\ref{outerplanar} we can assume that the outerplanar graph is bipartite.
Now, if we consider an auxiliary directed graph obtained from this 
outerplanar graph 
by putting a directed edge starting from a vertex and ending in  another vertex
after following a matching edge starting at the vertex and then a non-matching 
edge from there, then,
this auxiliary graph will have a directed cycle if and only if the matching
we start with is not unique. It is possible to show that the auxiliary graph is 
outerplanar.  Finally, since the reachability in directed outerplanar graphs 
is in Logspace (\cite{abc06}), we have that \UPM\ in outerplanar graphs is in 
Logspace.
\end{proof}

\section {Discussion}
We saw in Section \ref{isolation} that a perfect matching in bipartite 
planar graphs 
can be isolated by assigning small weights to the edges.
In this section we discuss the possibility of generalizing
this result in two orthogonal directions. 
For non-bipartite planar graphs and for bipartite but non-planar graphs.
The motivation is to isolate a perfect matching in a graph by 
having non-zero circulation on it.

\subsection {Non-bipartite Planar Matching}
Though non-bipartiteness is an issue, we can get rid of it to
a certain extent, though as expected, not completely .
\begin{lemma} Perfect Matching problem in general planar graphs reduces
to that of almost grid graphs. 
\end{lemma}

Now it suffices to get a non-vanishing circulations in almost grid graphs
to solve planar matching question. 
Unfortunately we don't yet know any way of achieving this though we have
some observations which might be helpful.

\begin{figure}
\includegraphics[scale=0.25] {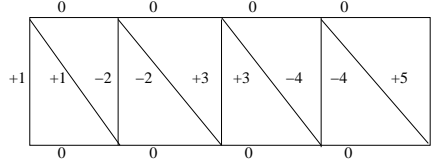}
\caption {Non-vanishing circulation in a non-bipartite graph}
\label{onerow}
\end{figure}
 
\begin{lemma}
One can have non-zero circulations for all the even cycles in the
graph in the Figure~\ref{onerow}. (The graph is simply one row of the grid with 
diagonals.)
\end{lemma}
\begin{proof} Observe that any even cycle in such a graph will either fall
in the grid or will fall in the grid formed by horizontal edge together with
diagonal edges.
Now, its easy to assign the weights as shown in the Figure~\ref{onerow} so that
all the horizontal edges get weight $0$ while vertical and diagonal edges get
weights so that any cycle in vertical or diagonal grid has non-vanishing 
circulation.
\end{proof}


In summary, non-bipartiteness seems to be an issue which is difficult to
get around. Hence, we keep the bipartiteness and next we explore the
possibility of generalizing our result for non-planar graphs.

\subsection {Bipartite Non-planar Matching}
Instead of looking at two dimensional grid we now consider three dimensional 
grids. The matching problem remains as hard as that for general 
bipartite graphs in this restriction as well.
\begin {lemma} One can embed any bipartite graph
in a three dimensional grid while preserving matchings.
\end {lemma}
\begin{proof} Firstly, one can make the degree of the graph
bounded by $3.$ Then one can use the third dimension to make the space for
crossings in the graph.
\end{proof}

\begin {oprob} Is the perfect matching problem for 
subgraphs of a three dimensional grid of height $2$ 
(constant height in general) in \NC\ ?
\end {oprob} 

The deterministic isolation of perfect matching is possible through
non-vanishing circulations as seen in Section \ref{isolation}.

\begin{oprob} Is small (log bit) weight non-vanishing circulation possible in 
every bipartite graph?
\end{oprob}

\subsection {Other Variations}
We know that 
the reachability in directed planar graphs reduces to bipartite planar 
matching while the reachability in layered grid graphs reduces to
the \UPM\ question in the same \cite{dkl07}.
Note that the isolation step in our algorithm works in Logspace. 
Solving the perfect matching question in bipartite planar graphs in Logspace
might be too strong to expect but at least one can ask the question
about \UPM\, which would put layered grid graph reachability in Logspace
or giving an orthogonal bound for the same.
\begin{oprob}
Is bipartite planar \UPM\ in \Log ?
\end{oprob}

We saw how to isolate a perfect matching in bipartite planar graph.
The isolation holds for maximum matching in bipartite planar graphs.
However, we do not know how to extract out the maximum weight perfect
matching in \NC.
\begin{oprob}
Is it possible to extract out the isolated maximum matching 
in \NC\  even for bipartite planar graphs?
\end{oprob}
Finally, the question of constructing a perfect matching in planar graphs
in \NC\ still remains open. 

\section {Acknowledgements}
We thank Meena Mahajan for useful discussion and comments on the preprint.

\end{document}